\documentclass[runningheads]{llncs}
\usepackage{color}
\usepackage{graphicx}

\usepackage{amsmath}
\usepackage{amsthm}
\usepackage{bm}
\usepackage{thmtools}
\usepackage{thm-restate}
\usepackage{hyperref}

\DeclareMathOperator{\bin}{bin}
\DeclareMathOperator{\enc}{enc}

\title{The Complexity of Approximate Pattern Matching on De Bruijn Graphs\thanks{This research is supported in part by the U.S. National Science Foundation (NSF) grants CCF-1704552, CCF-1816027, and CCF-2112643.}}

\author{Daniel Gibney${}^1$, Sharma V. Thankachan${}^2$, and Srinivas Aluru${}^1$}
\institute{${}^1$ School of Computational Science and Engineering, Georgia Institute of Technology, Atlanta, USA\\ 
${}^2$ Department of Computer Science, University of Central Florida, Orlando, USA\\
\email{Email: \{daniel.j.gibney,sharma.thankachan\}@gmail.com, aluru@cc.gatech.edu}\vspace*{-0.15in}}

\begin{document}
\maketitle
\begin{abstract}
Aligning a sequence to a \emph{walk} in a labeled graph is a problem of fundamental importance to Computational Biology. For finding a walk in an arbitrary graph with $|E|$ edges that exactly matches a pattern of length $m$, a lower bound based on the Strong Exponential Time Hypothesis (SETH) implies an algorithm significantly faster than $\mathcal{O}(|E|m)$ time is unlikely [Equi \emph{et al.}, ICALP 2019]. However, for many special graphs, such as de Bruijn graphs, the problem can be solved in linear time  [Bowe \emph{et al.}, WABI 2012]. For approximate matching, the picture is more complex. When edits (substitutions, insertions, and deletions) are only allowed to the pattern, or when the graph is acyclic, the problem is again solvable in $\mathcal{O}(|E|m)$ time. When edits are allowed to arbitrary cyclic graphs, the problem becomes NP-complete, even on binary alphabets  [Jain \emph{et al.}, RECOMB 2019]. These results hold even when edits are restricted to only substitutions. Despite the popularity of de Bruijn graphs in Computational Biology, the complexity of approximate pattern matching on de Bruijn graphs remained open. We investigate this problem and show that the properties that make de Bruijn graphs amenable to efficient exact pattern matching do not extend to approximate matching, even when restricted to the substitutions only case {with alphabet size four}.
Specifically, we prove that determining the existence of a matching walk in a de Bruijn graph is NP-complete when substitutions are allowed to the graph. In addition, we demonstrate that an algorithm significantly faster than $\mathcal{O}(|E|m)$ is unlikely for de Bruijn graphs in the case where only substitutions are allowed to the pattern. This stands in contrast to pattern-to-text matching where exact matching is solvable in linear time, like on de Bruijn graphs, but approximate matching under substitutions is solvable in subquadratic $\Tilde{O}(n\sqrt{m})$ time, where $n$ is the text's length  [Abrahamson, SIAM J. Computing 1987].
\end{abstract}

\section{Introduction}

De Bruijn graphs are an essential tool in Computational Biology. Their role in de novo assembly spans back to the 1980s~\cite{pevzner19891}, and their application in assembly has been researched extensively since then~\cite{DBLP:journals/jcb/ChikhiLJSM15,DBLP:journals/almob/ChikhiR13,DBLP:conf/sc/GeorganasBCORY14,DBLP:conf/recomb/LinSYCP16,DBLP:conf/recomb/PengLYC10,DBLP:journals/bioinformatics/PengLYLZC13,ren2012evaluating,zerbino2008velvet}. More recently, de Bruijn graphs have been applied in metagenomics and in the representation of large collections of genomes~\cite{DBLP:journals/pc/FlickJPA17,kamal2017bruijn,DBLP:journals/bioinformatics/LiLLSL15,DBLP:journals/pnas/PellHCHTB12,DBLP:journals/bioinformatics/YeT16} and for solving other problems such as read-error correction~\cite{DBLP:journals/bioinformatics/LimassetFP20,DBLP:journals/bioinformatics/MorisseLL18} and compression~\cite{BenoitLLDDUR15,HolleyWSH18}. Due to the popularity of de Bruijn graphs in the modeling of sequencing data, an algorithm to efficiently find \emph{walks} in a de Bruijn graph matching (or approximately matching) a given query pattern would be a significant advancement. For example, in metagenomics, such an algorithm could quickly detect the presence of a particular species within genetic material obtained from an environmental sample. Or, in the case of read-error correction, such an algorithm could be used to efficiently find the best mapping of reads onto a `cleaned' reference de Bruijn graph with low-frequency k-mers removed~\cite{DBLP:journals/bioinformatics/LimassetFP20}. To facilitate such tasks, several algorithms (often seed-and-extend type heuristics) and software tools have been developed that perform pattern matching on de Bruijn (and general) graphs~\cite{DBLP:journals/bioinformatics/AlmodaresiSSP18,DBLP:journals/bmcbi/HeydariMPF18,holley2012blastgraph,DBLP:journals/jcb/KavyaTSS19,DBLP:journals/bmcbi/LimassetCRP16,DBLP:journals/bioinformatics/LiuGBW16,DBLP:journals/tcs/Navarro00,rautiainen2017aligning}.

The importance of pattern matching on labeled graphs in Computational Biology and other fields has caused a recent surge of interest in the theoretical aspects of this problem. In turn, this has led to many new fascinating algorithmic and computational complexity results. However, even with this improved understanding of the theory of pattern matching on labeled graphs, our knowledge is still lacking in many respects concerning specific, yet extremely relevant, graph classes.
An overview of the current state of knowledge is provided in Table \ref{table:complexity_overview}.

\begin{table}[ht]
\small
\centering
\begin{tabular}{| l | l |  l |} 
\hline
&  Exact Matching &  Approximate Matching\\ 
\hline
& \underline{Solvable in Linear Time} & \underline{Solvable in $\mathcal{O}(|E|m)$ time}  \\ 
 & $\bullet$ Wheeler Graphs~\cite{DBLP:journals/tcs/GagieMS17} & $\bullet$ DAGs: Substitutions/Edits to graph~\cite{DBLP:journals/jcb/KavyaTSS19} \\ 
Easy & ~(e.g. de Bruijn graphs, &  $\bullet$ General graphs:\\
&NFAs for multiple strings) &~~Substitutions/Edits to pattern~\cite{DBLP:journals/jal/AmirLL00}\\
& & $\bullet$ de Bruijn Graphs: Substitutions to pattern \\
& & ~~-No strongly Sub-$\mathcal{O}(|E|m)$ alg. (this paper) \\
\hline
& \underline{NO Strongly Sub-$\mathcal{O}(|E|m)$ Alg.}~ & \underline{NP-Complete} \\
&  $\bullet$ General graphs ~\cite{DBLP:conf/icalp/EquiGMT19,DBLP:conf/sosa/GibneyHT21} & $\bullet$ General graphs:\\
& ~(including DAGs with & ~~Substitutions/Edits to vertex labels~\cite{DBLP:journals/jal/AmirLL00,DBLP:conf/recomb/JainZGA19} \\ 
Hard & ~~total degree $\leq$ 3) & $\bullet$ de Bruijn Graphs:\\
&  & ~~Substitutions to vertex labels (this paper)\\
\hline
\end{tabular}
\caption{The computational complexity of pattern matching on labeled graphs}
\label{table:complexity_overview}
\vspace{-2em}
\end{table}

For general graphs, we can consider exact and approximate matching. For exact matching, conditional lower-bounds based on the Strong Exponential Time Hypothesis (SETH), and other conjectures in circuit complexity, indicate that an $\mathcal{O}(|E|m^{1-\varepsilon} + |E|^{1-\varepsilon}m)$ time algorithm with any constant $\varepsilon > 0$, for a graph with $|E|$ edges and a pattern of length $m$, is highly unlikely (as is the ability to shave more than a constant number of logarithmic factors from the $\mathcal{O}(|E|m)$ time complexity)~\cite{DBLP:conf/icalp/EquiGMT19,DBLP:conf/sosa/GibneyHT21}. These results hold for even very restricted types of graphs, for example, DAGs with maximum total degree three and binary alphabets. For approximate matching, when edits are only allowed in the pattern, the problem is solvable in $\mathcal{O}(|E|m)$ time~\cite{DBLP:journals/jal/AmirLL00}. If edits are also permitted in the graph, but the graph is a DAG, matching can be done in the same time complexity~\cite{DBLP:journals/jcb/KavyaTSS19}. However, the problem becomes NP-complete when edits are allowed in arbitrary cyclic graphs. This was originally proven in \cite{DBLP:journals/jal/AmirLL00} for large alphabets and more recently proven for binary alphabets in \cite{DBLP:conf/recomb/JainZGA19}. These results hold even when edits are restricted to only substitutions. The distinction between modifications to the graph and modifications to the pattern is important as these two problems are fundamentally different. When changes are made to cyclic graphs the same modification can be encountered multiple times while matching a pattern with no additional cost (see Section 3.1 in \cite{DBLP:conf/recomb/JainZGA19} for a detailed discussion). Furthermore, algorithmic solutions appearing in \cite{DBLP:journals/jcb/KavyaTSS19,DBLP:journals/tcs/Navarro00,rautiainen2017aligning} are for the case where modifications are performed only to the pattern.

De Bruijn graphs are  interesting  from a theoretical perspective. Many graphs allow for extending Burrows-Wheeler Transformation (BWT) based techniques for efficient pattern matching. Sufficient conditions for doing this are captured by the definition of Wheeler graphs, introduced in \cite{DBLP:journals/tcs/GagieMS17}, and further studied in \cite{DBLP:journals/corr/abs-2002-10303,DBLP:conf/dcc/AlankoGNB19,DBLP:journals/corr/abs-2009-03675,DBLP:journals/corr/abs-2101-12341,DBLP:conf/esa/GibneyT19}. De Bruijn graphs are themselves Wheeler graphs, hence on a de Bruijn graph exact pattern matching  is solvable in linear time.
However, the complexity of approximate matching in de Bruijn graphs when permitting modifications to the graph or modifications to the pattern remained open~\cite{DBLP:conf/recomb/JainZGA19}.

We make two important contributions (see Table \ref{table:complexity_overview}). First, we prove that for de Bruijn graphs, despite exact matching being solvable in linear time, the approximate matching problem with vertex label substitutions is NP-complete. Second, we prove that a strongly subquadratic time algorithm for the approximate pattern matching problem on de Bruijn graphs, where substitutions are only allowed in the pattern, is not possible under SETH. This confirms the optimality of the known quadratic time algorithms when considering polynomial factors. 
To the best of our knowledge, these are the first such results for any type of Wheeler graph. 
Note that pattern-to-text matching (under substitutions) can be solved in sub-quadratic  
$\tilde{\mathcal{O}}(n\sqrt{m})$ time, where $n$ is the text's length~\cite{DBLP:journals/siamcomp/Abrahamson87}. 

\subsection{Technical Background and Our Results}
\label{sec:background_and_results}


\textbf{Notation for edges:} For a directed edge from a vertex $u$ to a vertex $v$ we will use the notation $(u,v)$. Additionally, we will refer to $u$ as the \emph{tail} of $(u,v)$, and $v$ as the \emph{head} of $(u,v)$.

\textbf{Walks versus paths:} A distinction must be made between the concept of a \emph{walk} and a \emph{path} in a graph. A walk is a sequence of vertices $v_{1}$, $v_{2}$, ..., $v_{t}$ such that for each $i \in [1,t-1]$, $(v_{i}, v_{i + 1}) \in E$. Vertices can be repeated in a walk. A path is a walk where vertices are not repeated. The length of a walk is defined as the number of edges in the walk, $t-1$, or equivalently one less than the number of vertices in the sequence (counted with multiplicity).  
This work will be concerning the existence of walks.

\textbf{Induced subgraphs:} An induced subgraph of a graph $G = (V, E)$ consists of a subset of vertices $V' \subseteq V$, and all edges $(u,v) \in E$ such that $u,v \in V'$. This is in contrast to an arbitrary subgraph of $G$, where an edge can be omitted from the subgraph, even if both of its incident vertices are included.

\textbf{De Bruijn graphs:}
An \emph{order-$k$} full de Bruijn graph is a compact representation of all $k$-mers (strings of length $k$) from an alphabet $\Sigma$ of size $\sigma$. It consists of $\sigma^k$ vertices, each corresponding to a unique $k$-mer (which we call as its \emph{implicit label}) in $\Sigma^k$.
There is a directed edge from each vertex with implicit label $s_1s_2...s_k \in \Sigma^k$ to the $\sigma$ vertices with implicit labels $s_2s_3...s_{k}\alpha$, $\alpha \in \Sigma$.
We will work with induced subgraphs of full de Bruijn graphs in this paper. 
We assign to every vertex $v$ a label $L(v) \in \Sigma$, such that the implicit label of $v$ is $L(u_1)L(u_2)...L(u_{k-1})L(v)$ where $u_1, u_2,..., u_{k-1}, v$ is any length $k-1$ walk ending at $v$. This is equivalent to the notion of a de Bruijn graph constructed from $k$-mers commonly used in Computational Biology.



\textbf{Strings and Matching:} For a string $S$ of length $n$ indexed from $1$ to $n$, we use $S[i]$ to denote the $i^{th}$ symbol in $S$. We use $S[i,j]$ to denote the substring $S[i]S[i+1]...S[j]$. If $j < i$, then we take $S[i,j]$ as the empty string. As mentioned above, we will consider every vertex $v$ as labeled with a single symbol $L(v) \in \Sigma$. 
A pattern $P[1,m]$ matches a walk $v_{1}$, $v_{2}$, ..., $v_{m}$ iff $P[i] = L(v_i)$ for every $i \in [1,m]$.

With these definitions in hand, we can formally define our first problem. 

\begin{problem}[Approximate matching with vertex label substitutions]
\label{prob:sub}
Given a vertex labeled graph $D = (V,E)$ with alphabet $\Sigma$ of size $\sigma$, pattern $P[1,m]$, and integer $\delta \geq 0$, determine if there exists a walk in $D$ matching $P$ after at most $\delta$ substitutions to the vertex labels.
\end{problem}

\begin{theorem}
\label{thm:NPC_hardness}
Problem \ref{prob:sub} is NP-complete on de Bruijn graphs with  $\sigma = 4$.
\end{theorem}
Theorem \ref{thm:NPC_hardness} is proven in Section \ref{sec:NPC_hardness}.  Intuitively, our reduction transforms a general directed graph into a de Bruijn that maintains key topological properties related to the existence of walks. 
The distinct problem of approximately matching a pattern to a \emph{path} in a de Bruijn graph was shown to be NP-complete in \cite{DBLP:journals/bmcbi/LimassetCRP16}. As mentioned by the authors of that work, the techniques used there do not appear to be easily adaptable to the problem for walks. Our approach uses edge transformations more closely inspired by those used in \cite{DBLP:conf/wabi/KapunT13} for proving hardness on the paired de Bruijn sound cycle problem.

\begin{problem}[Approximate matching with substitutions within the pattern]
\label{prob:sub_pattern}
Given a vertex labeled graph $D = (V,E)$ with alphabet $\Sigma$ of size $\sigma$, pattern $P[1,m]$, and integer $\delta \geq 0$, determine if there exists a walk in $D$ matching $P$ after at most $\delta$ substitutions to the symbols in $P$.
\end{problem}

For Problem \ref{prob:sub_pattern} we provide a hardness result based on SETH, which is frequently used for establishing
conditional optimality of polynomial time algorithms~\cite{DBLP:journals/talg/AbboudBHWZ18,DBLP:conf/focs/BackursI16,DBLP:conf/icalp/EquiGMT19,DBLP:conf/spire/Gibney20,DBLP:conf/sosa/GibneyHT21,DBLP:conf/esa/HoppenworthBGT20}.  We refer the reader to \cite{DBLP:conf/iwpec/Williams15} for the definition of SETH and for the reduction to the Orthogonal Vectors problem (OV), which is utilized to prove Theorem \ref{thm:SETH_hardness}. 

\begin{theorem}
\label{thm:SETH_hardness}
Conditioned on SETH, for all constants $\varepsilon > 0$, there does not exist an $\mathcal{O}(|E|m^{1-\varepsilon} + |E|^{1-\varepsilon}m)$ time algorithm for Problem \ref{prob:sub_pattern} on de Bruijn graphs with $\sigma = 4$. 
\end{theorem}

Note that the order of the de Bruijn graphs used in ours proofs are $\Theta(\log^2 |V|)$ for Theorem \ref{thm:NPC_hardness} and  $\Theta(\log |V|)$ for Theorem \ref{thm:SETH_hardness}.

\section{NP-Completeness of Problem 1 on De Bruijn Graphs}
\label{sec:NPC_hardness}
Our proof of NP-completeness uses a reduction from the Hamiltonian Cycle Problem on directed graphs, which is the problem of deciding if there exists a cycle through a directed graph that visits every vertex exactly once. It was proven NP-complete even when restricted 
to directed graphs where the number of edges is linear in the number of vertices~\cite{DBLP:journals/ipl/Plesnik79}.
To present the reduction, we introduce the concept of \emph{merging} two vertices. To merge vertices $u$ and $v$, we create a new vertex $w$. We then take all edges with either $u$ or $v$ as their head and make $w$ their new head. Next, we take all edges with either $u$ or $v$ as their tail and make $w$ their new tail. This makes the edges $(u,v)$ and $(v,u)$ (if they existed) into self-loops for $w$. If two self-loops are formed, we delete one of them. Finally, we delete the original vertices $u$ and $v$.

\begin{figure}
\centering
\begin{minipage}{.48\textwidth}
    \centering
    \includegraphics[width=.9\textwidth]{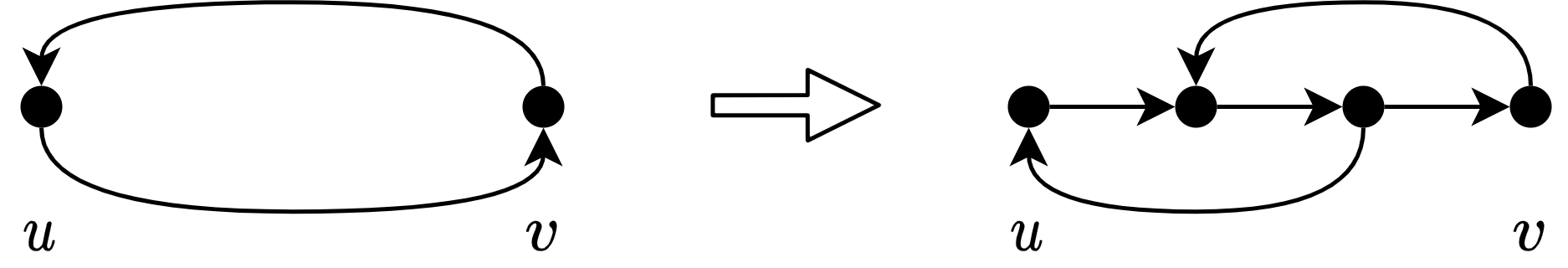}
     \caption{Gadget to remove cycles of length $2$ from the initial input graph.}
    \label{fig:cycle_gadget}
\end{minipage}
\begin{minipage}{.48\textwidth}
    \centering
    \includegraphics[width=.7\textwidth]{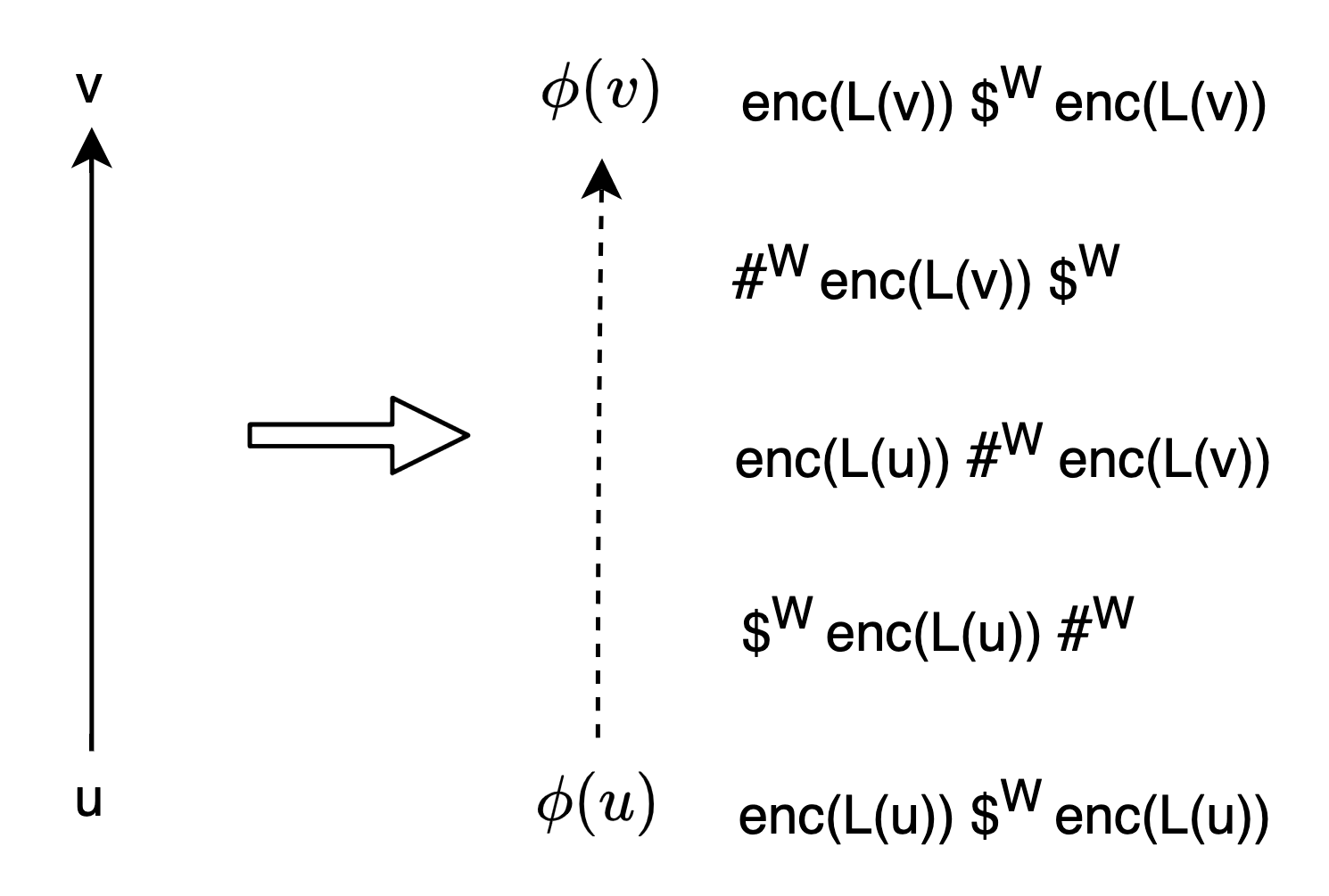}
    \vspace{-1em}
    \caption{The transformation from edges to paths used in our reduction.}
    \label{fig:edge_transform}
\end{minipage}
\vspace{-1em}
\end{figure}

\begin{figure}
\centering
\includegraphics[width=.9\textwidth]{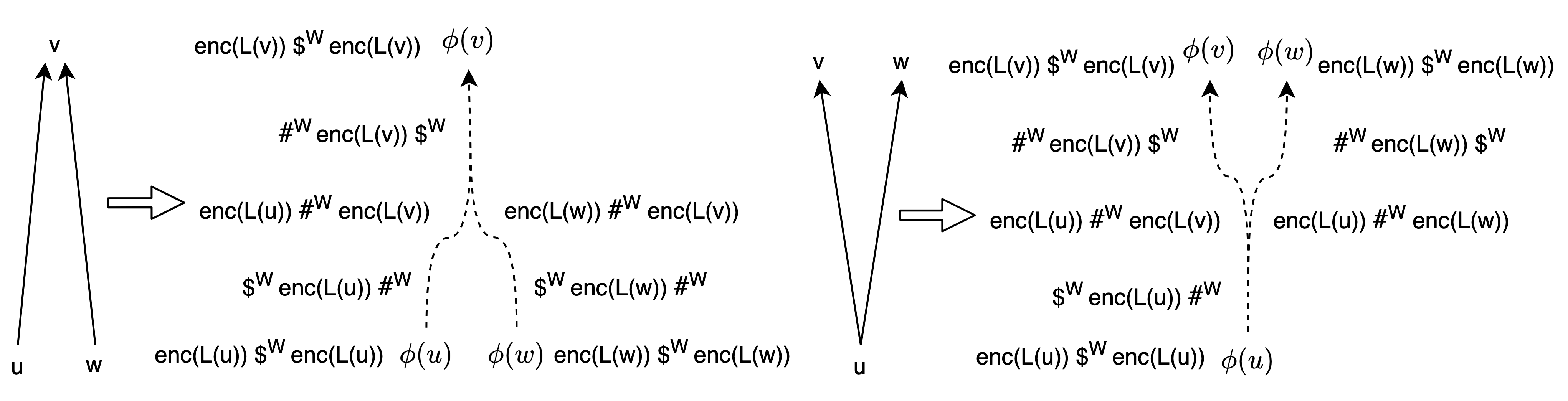}
\vspace{-1em}
\caption{Vertices with the same implicit label are merged while transforming $D$ to $D'$, causing edges with shared vertices to become paths with shared vertices.}
\label{fig:merge}
\vspace{-1em}
\end{figure}

\vspace{-2em}
\subsection{Reduction}

We start with an instance of the Hamiltonian cycle problem on a directed graph where the number of edges is linear in the number of vertices. 
We can assume there are no self-loops or vertices with in-degree or out-degree zero. To simplify the proof, we first eliminate any cycles of length $2$ using the gadget in Figure \ref{fig:cycle_gadget}. We denote the resulting graph as $D = (V,E)$ and let $n = |V|$.We assign each vertex $v \in V$ a unique integer $L(v) \in [0,n-1]$. Let $\ell = \lceil \log n \rceil$, $\bin(i)$ be the standard binary encoding of $i$ using $\ell$ bits and $\Sigma = \{\$, \#, 0, 1\}$. Define $\enc(i) = (0^{2\ell} 1)^{2\ell} \bin(i)$, $W = |\enc(i)|$, and $k = 3W$.

We construct a new (de Bruijn) graph $D' = (V', E')$ as follows:  Initially $D'$ is the empty graph. For $i = 0, 1, \dots, n-1$, for each edge $(u, v) \in E$ where $L(v) = i$, create a new path whose concatenation of vertex labels is $\#^W\!\enc(i)\$^W\!\enc(i)$. The vertex $u$ will correspond with a new vertex $\phi(u)$ at the start of this path, and the vertex $v$ will correspond with a new vertex $\phi(v)$ at the end of this path.  The vertex $\phi(v)$ has the implicit label $\enc(L(v))\$^W\!\enc(L(v))$. The vertex $\phi(u)$ is temporarily assigned the implicit label $\enc(L(u))\$^W\!\enc(L(u))$. See Figure \ref{fig:edge_transform}. We call vertices with implicit labels of the form $\enc(L(\cdot))\$^W\!\enc(L(\cdot))$ \emph{marked vertices}. We use the notation $\phi((u,v))$ to denote the path created when applying this transformation to $(u,v) \in E$. After the path $\phi((u,v))$ is created,  vertices in $V'$ having the same implicit label are merged, and parallel edges are deleted (Figure \ref{fig:merge}). See Figure \ref{fig:transformation} for a complete example. Finally, let $\delta = 2\ell(n-1)$ and
\begin{align*}
P = &\#^W\!\enc(0)\$^W\!\enc(0)\#^W\! \enc(1)\$^W\!\enc(1)\#^W\! \hdots\\ 
&\#^W\! \enc(n-1)\$^W\!\enc(n-1)\#^W\!\enc(0)\$^W\!\enc(0).
\end{align*}
We will show that there exists a walk in $D'$ matching $P$ with at most $\delta$ vertex label substitutions iff $D$ contains a Hamiltonian cycle.

\begin{figure}
    \begin{minipage}{0.3\textwidth}
    \centering
    \includegraphics[width=\textwidth]{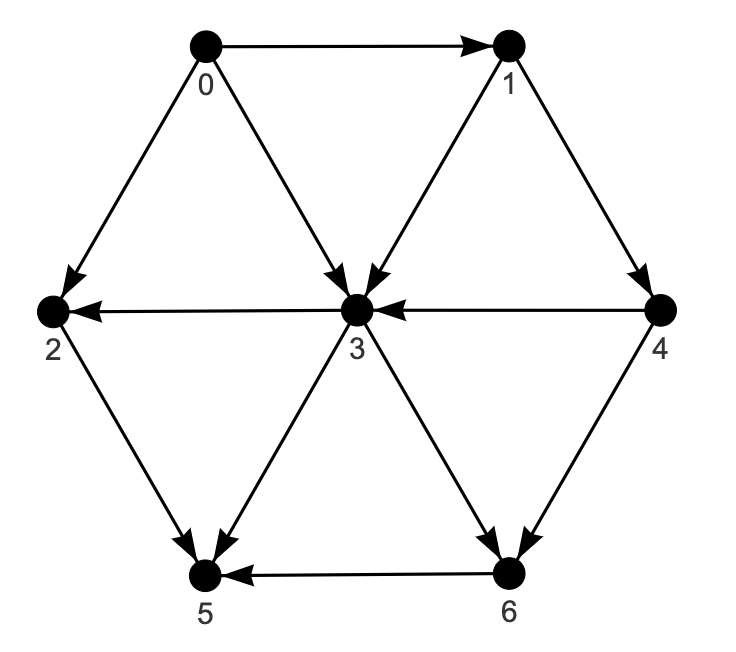}
    \end{minipage}
    \begin{minipage}{0.7\textwidth}
    \centering
    \includegraphics[width=\textwidth]{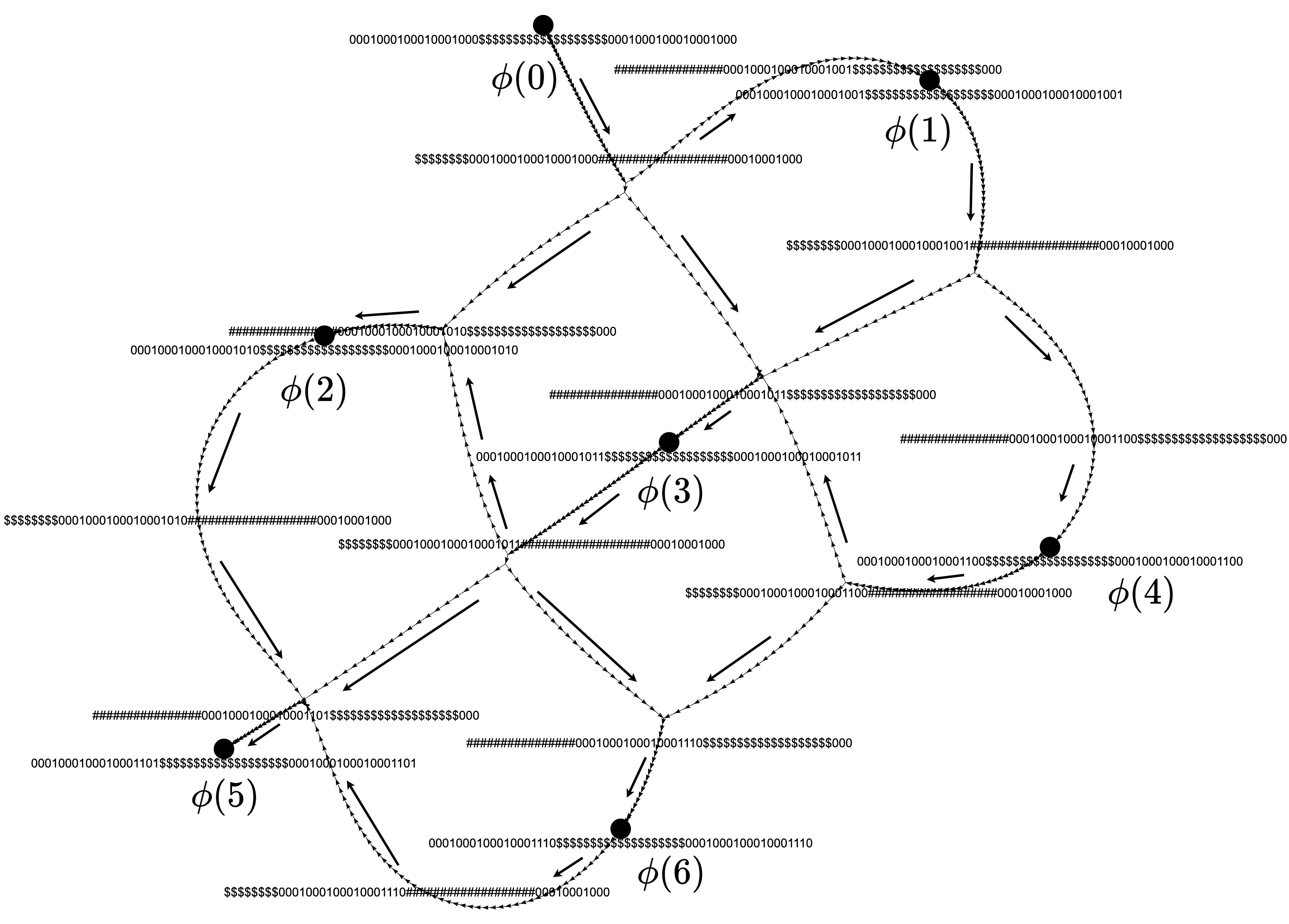}
    \end{minipage}
    \caption{(Left) A graph before the reduction is applied to it. (Right) The transformed graph. A subset of the implicit labels are shown, and the path directions are annotated by arrows beside each path. Note that $\enc(\cdot)$ has been modified to have the prefix $(0^\ell 1)^{\ell + 1}$ so that it fits in the figure. Also, unlike in the figure, we assume in our reduction that there are no vertices with in-degree or out-degree zero.}
    \label{fig:transformation}
    \vspace{-1em}
\end{figure}

\subsubsection{Proof of Correctness}
\label{sec:NPC_proof_of_correctness}

\begin{restatable}{lemma}{npcIsDebruijn}
\label{lem:is_debruijn}
The graph $D'$ constructed as above is a de Bruijn graph.
\end{restatable}

\begin{proof} \textit{(Overview)}
Three properties must be proven:  (i) Implicit labels are unique, meaning for every implicit label at most one vertex is assigned that label; (ii) No edges are missing, i.e., if the implicit label of $y \in V'$ is $S \alpha$ for some string $S[1,k-1]$ and symbol $\alpha \in \Sigma$, and there exists a vertex $x \in V'$ with implicit label $\beta S[1,k-1]$ for some symbol $\beta \in \Sigma$, then $(x,y) \in E'$; (iii) Implicit labels are well-defined, in that every walk of length $k-1$ ending at a vertex $x \in V'$ matches the same string (the implicit label of $x$); The most involved of these is proving property (ii), which requires analyzing several cases. The full proof is given in Appendix \ref{sec:appendix_npc}.
\end{proof}

The correctness of the reduction remains to be shown. Lemmas \ref{lem:walk_length}-\ref{lem:path_iff_in_D} establish useful structural properties of $D'$, Lemma \ref{lem:hampath_implies_walk} proves that the existence of a Hamiltonian Cycle in $D$ implies an approximate matching in $D'$, and Lemmas \ref{lem:dollar_signs}-\ref{lem:match_implies_hampath} demonstrate the converse.

\begin{restatable}{lemma}{npcWalkLength}
\label{lem:walk_length}
Any walk between two marked vertices $\phi(u)$ and $\phi(v)$ 
containing no additional marked vertices has length $4W$. Hence, we can conclude any such walk is a path.
\end{restatable}

\begin{proof} \textit{(Overview)}
This is proven using induction on the number of edges transformed. It is shown that for every vertex, a key property regarding the distances to its closest marked vertices continues to hold after vertices on any newly created path are merged. See Appendix \ref{sec:appendix_npc} for the full proof.
\end{proof}

\begin{lemma}
\label{lem:disjoint_paths}
For $(u_1,v_1), (u_2, v_2) \in E$, unless $u_1 = u_2$ or $v_1 = v_2$, $\phi((u_1,v_1))$ and $\phi((u_2, v_2))$ share no vertices. 
\end{lemma}

\begin{proof}
In the case where $\{u_1, v_1\} \cap \{u_2, v_2\} = \emptyset$ (Figure \ref{fig:no_intersection} left), every implicit vertex label in $\phi((u_1,v_1))$ contains $\enc(L(u_1))$ or $\enc(L(v_1))$ (or both), and contains neither $\enc(L(u_2))$ nor $\enc(L(v_2))$. Similarly, every implicit vertex label in $\phi((u_2, v_2))$ contains $\enc(L(u_2))$ or $\enc(L(v_2))$ (or both) and contains neither $\enc(L(u_1))$ nor $\enc(L(v_1))$. This implies that none of the implicit labels match between the two paths, thus no vertices are merged. In the case where $v_1 = u_2$ and $u_1 \neq v_2$ (Figure \ref{fig:no_intersection}, right), the implicit labels of vertices $\phi((u_1, v_1))$ not containing $\enc(L(u_1))$ have $\#$ symbols in different positions than implicit labels of vertices in $\phi((u_2, v_2))$ not containing $\enc(L(v_2))$, and, since $v_1 \neq v_2$, cannot match the implicit labels of vertices in $\phi((u_2, v_2))$ containing $\enc(L(v_2))$. Vertices in $\phi((u_1, v_1))$ with implicit labels containing $\enc(L(u_1))$ have $\#$ symbols in different positions than implicit labels of vertices in $\phi((u_2, v_2))$ not containing $\enc(L(u_2))$, and, since $u_1 \neq u_2$, cannot match the implicit labels of vertices in $\phi((u_2, v_2))$ containing $\enc(L(u_2))$. The case $u_1 = v_2$ and $u_2 \neq v_1$ is symmetric. The case $u_1 = v_2$ and $v_1 = u_2$ cannot happen since, by the use of our gadget in Figure \ref{fig:cycle_gadget}, $D$ cannot contain the edges $(u_1,v_1)$ and $(v_1,u_1)$.
\end{proof}

\begin{figure}
\includegraphics[width=\textwidth]{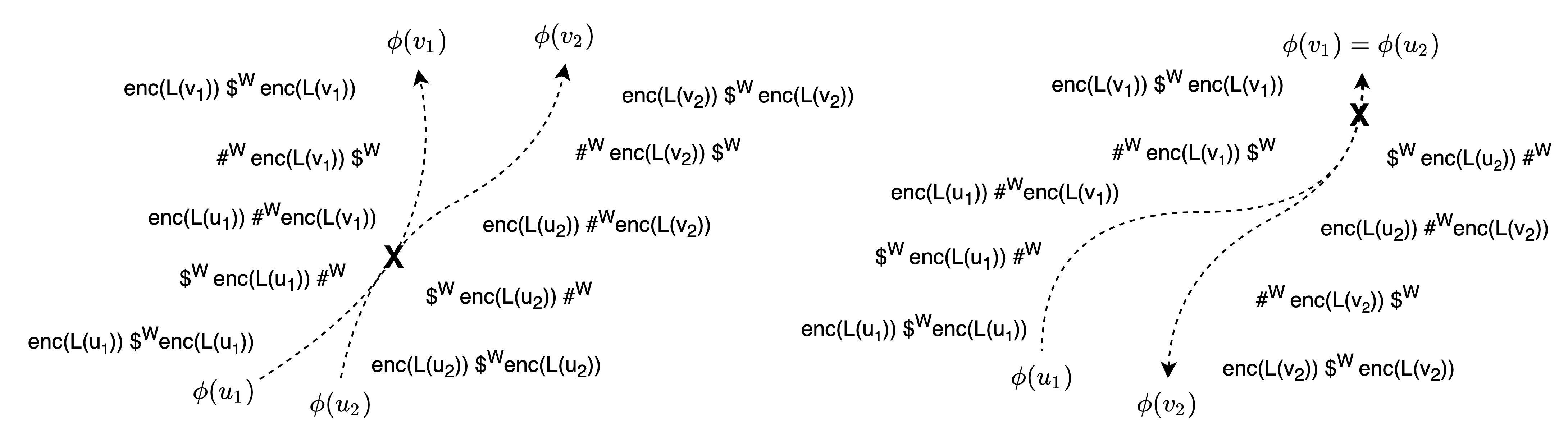}
\vspace{-1em}
\caption{Examples where paths between marked vertex cannot share any vertex: (Left) The case where $\{u_1, v_1\}\cap \{u_2, v_2\} = \emptyset$. (Right) The case where $v_1 = u_2$ and $u_1 \neq v_2$.
}
\label{fig:no_intersection}
\vspace{-2em}
\end{figure}

\begin{restatable}{lemma}{npcPathIffInD}
\label{lem:path_iff_in_D}
There exists a path from a marked vertex $\phi(u) \in V'$ to a marked vertex $\phi(v) \in V'$ containing no other marked vertices iff there is an edge $(u,v) \in E$.
\end{restatable}

\begin{proof} \textit{(Overview)}
It is clear from construction that if $(u,v) \in D$, then such a path exists in $D'$. In the other direction, we utilize Lemmas \ref{lem:walk_length} and \ref{lem:disjoint_paths} to show that such a path existing without a corresponding edge would create a contradiction. The full proof is provided in Appendix \ref{sec:appendix_npc}.
\end{proof}


\begin{lemma}
\label{lem:hampath_implies_walk}
If $D$ has a Hamiltonian cycle, then $P$ can be matched in $D'$ with at most $\delta$ substitutions to vertex labels of $D'$. 
\end{lemma}

\begin{proof}
To obtain a matching walk, follow the cycle corresponding to a solution in $D$ starting with the marked vertex in $V'$ corresponding to the vertex in $V$ with label $0$. By Lemma \ref{lem:path_iff_in_D}, each edge traversed in $D$ corresponds to a path in $D'$. While traversing these paths, modify the vertex labels in $D'$ corresponding to the substrings $\bin(i)$ to match $P$. Assuming no conflicting substitutions are needed, this requires at most $2\ell(n-1)$ substitutions. 

It remains to be shown that no conflicting label substitutions will be necessary. Consider the edges $(u_1,v_1), (u_2, v_2) \in E$ used in the Hamiltonian cycle in $D$. We will never have $u_1 = u_2$ or $v_1 = v_2$. Hence, by Lemma \ref{lem:disjoint_paths}, the sets of vertices on the paths $\phi((u_1, v_1))$ and $\phi((u_2, v_2))$ are disjoint.
\end{proof}

\begin{restatable}{lemma}{npcDollarSigns}
\label{lem:dollar_signs}
If $P$ can be matched in $D'$ with at most $\delta$ substitutions to vertex labels of $D'$, then all $\$$'s in $P$ are matched with non-substituted $\$$'s in $D'$ and all $\#$'s in $P$ are matched with non-substituted $\#$'s in $D'$. Consequently, we can assume the only substitutions are to the vertex labels corresponding to $\bin(i)$'s within $\enc(i)$'s.
\end{restatable}

\begin{proof} \textit{(Overview)}
We establish the existence of a long, non-branching path for every marked vertex that can be traversed at most once when matching $P$. This, combined with maximal paths of, $\$$, $\#$, and 0/1-symbols, all being of length $W$, makes it so that `shifting' $P$ to match a portion of $D$ forces the shift to occur throughout the walk traversed while matching $P$. Utilizing the large Hamming distance between shifted instances of two encodings, we can then show that not matching all non-0/1 symbols requires more than $\delta$ substitutions. The full proof is provided in Appendix \ref{sec:appendix_npc}.
\end{proof}

Post-substitution to vertex labels, we will refer to a vertex as marked if there exists a walk ending at it that matches a string of the form $\enc(L(u)) \$^W \enc(L(u))$, $u \in V$. Note that this definition does not require all length $k-1$ walks ending at such a vertex to match the same string.

\begin{lemma}
\label{lem:remains_marked}
If $P$ can be matched in $D'$ with at most $\delta$ substitutions to vertex labels of $D'$, then no additional marked vertices are created due to vertex substitutions.
\end{lemma}

\begin{proof}
Pre-substitution, only marked vertices have implicit labels of the form $S_1 \$^W S_2$ where $S_1$ and $S_2$ contain no $\$$ symbols. Hence, the only way that a vertex could have a walk ending at it that matches a pattern of that form post-substitution is if either it was originally a marked vertex, or some non-0/1-symbols were substituted in $D'$. However, by Lemma \ref{lem:dollar_signs} the latter case cannot happen, and only originally marked vertices have walks ending at them matching strings of the form $S_1 \$^W S_2$ post-substitution.
\end{proof}

\begin{lemma}
\label{lem:visit_once}
If $P$ can be matched in $D'$ with at most $\delta$ substitutions to vertex labels of $D'$, then each originally marked vertex in $D'$ is visited exactly once, except for an originally marked vertex at the end of a path matching $\enc(0)\$^W\enc(0)$ that is visited twice.
\end{lemma}

\begin{proof}
First, we show that all marked vertices, except the one with implicit label $\enc(0)\$^W\enc(0)$, are visited at most once. Pre-substitution, a marked vertex with implicit label $\enc(i) \$^W\!\enc(i)$ is at the end of a unique, branchless path of length $W$ matching $\enc(i)$. By Lemma \ref{lem:dollar_signs}, the only substitutions to this path made while matching $P$ are substitutions making it match $\enc(i')$, $i' \neq i$. If this path were modified to match $\enc(i')$, $i' > 0$, then the only way the marked vertex could be visited twice while matching $P$ is if after traversing the path, another path matching $\$^W$ is taken back to the start of this $\enc(i')$ path. However, any edges leaving this marked vertex are labeled with $\#$, making this impossible. By similar reasoning, the path matching $\enc(0)$ ending at a marked vertex is visited at most twice. We now show that each marked vertex is visited at least once. Suppose some marked vertex is not visited. By Lemma \ref{lem:remains_marked}, no additional marked vertices are created. Hence, a marked vertex ending a path matching $\enc(i)$, $i > 0$ is visited at least twice, or a marked vertex ending a path matching $\enc(0)$ is visited at least three times, a contradiction.
\end{proof}

\begin{lemma}
\label{lem:match_implies_hampath}
If $P$ can be matched in $D'$ with at most $\delta$ substitutions to vertex labels of $D'$, then $D$ has a Hamiltonian cycle.
\end{lemma}

\begin{proof}
By Lemma \ref{lem:path_iff_in_D}, the paths between marked vertices traversed while matching with $P$ correspond to edges between vertices in $D$. Combined with marked vertices being visited exactly once from Lemma \ref{lem:visit_once} (except the marked vertex ending a path matching $\enc(0)$), the walk matched by $P$ in $D'$ corresponds to a Hamiltonian cycle through $D$ beginning and ending at the vertex labeled $0$.
\end{proof}

This completes the proof of Theorem \ref{thm:NPC_hardness}. To see that $k = \Theta(\log^2 |V'|)$, first recall that $|V|$ is the number of vertices in the original graph, where we assumed $|E| = \mathcal{O}(|V|)$. At most $4W|E| = \mathcal{O}(k|V|)$ vertices are created in the reduction. Also, the proof of Lemma \ref{lem:dollar_signs} establishes that there is a unique set of at least $\Theta(k)$ vertices for every marked vertex, each one corresponding to a vertex in the original graph. Combining, we have that $|V'| = \Theta(k|V|)$. By construction, $k = \Theta(\log^2 |V|)$, and since $|V'| = \Theta(k|V|)$, $k = \Theta(\log^2 |V'|)$ as well.





\section{Hardness for Problem 2 on De Bruijn Graphs}
\label{sec:SETH_hardness}

\subsubsection{Reduction}

The Orthogonal Vectors Problem is defined as follows: given two sets of binary vectors $A, B \subseteq \{0,1\}^d$ where $|A| = |B| = N$, determine whether there exists vectors $a \in A$ and $b \in B$ such that their inner product is zero. Conditioned on SETH, a standard reduction shows that this cannot be solved in time $d^{\Theta(1)}N^{2 - \varepsilon}$ for any constant $\varepsilon > 0$~\cite{DBLP:conf/iwpec/Williams15}.

\begin{figure}
    \centering
    \includegraphics[width=\textwidth]{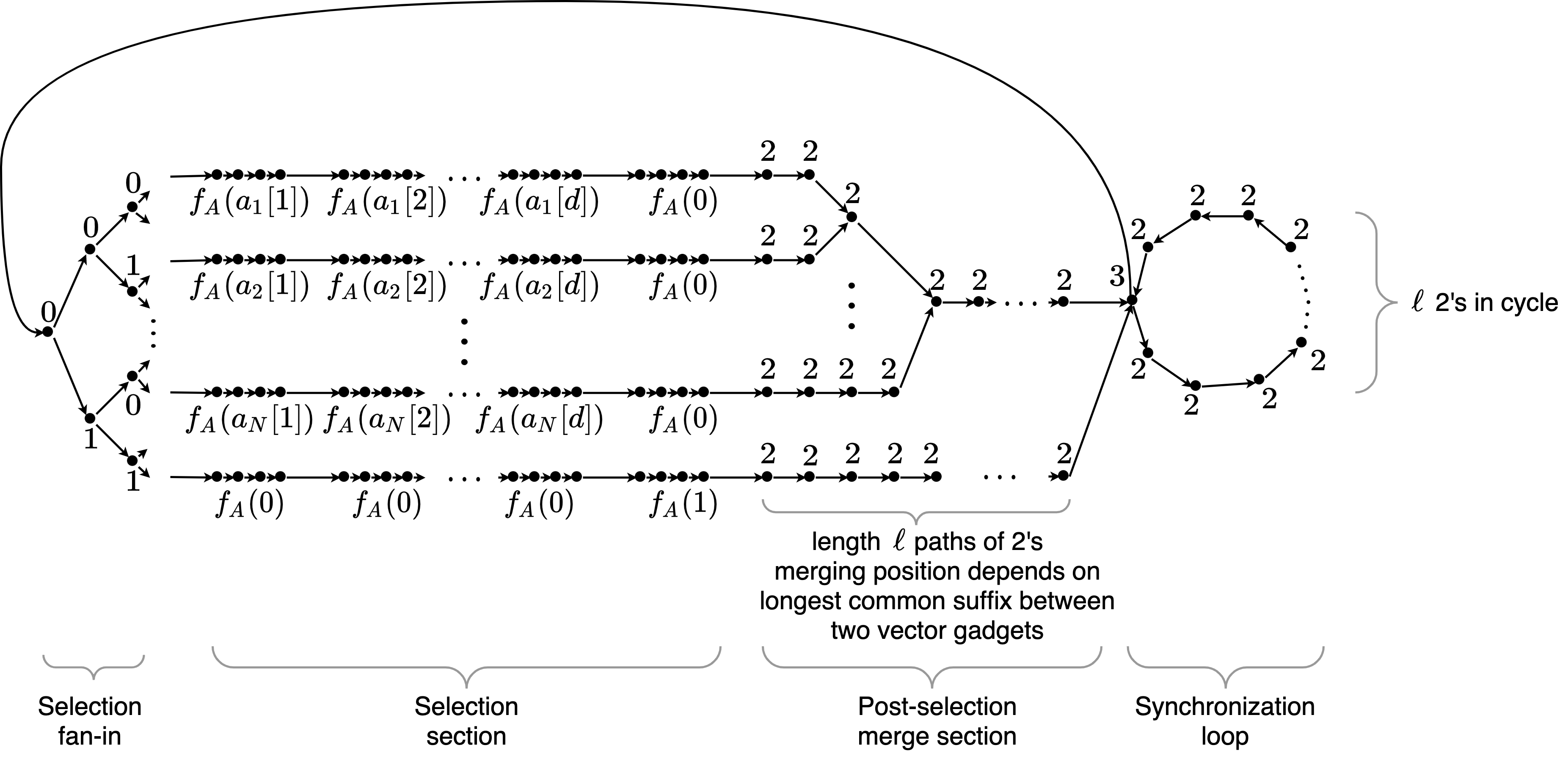}
    \caption{An illustration of the reduction from OV to Problem 2.}
    \label{fig:SETH_reduction}
    \vspace{-1.8em}
\end{figure}

Let the given instance of OV consist of $A, B \subseteq \{0,1\}^d$ where $|A| = |B| = N =  2^{m}$ for some natural number $m$. Hence, we have $\left\lceil \log (N+1) \right\rceil =  \log N + 1$. This will ease computation later. We also assume that $d > \log N$. This is reasonable, as if $d \leq \log N$, then $|A|$ and $|B|$ would contain either all vectors of length $d$ or repetitions. 

We will next provide a formal description of the graph $D$ our reduction creates from the set $A = \{a_1, a_2, ..., a_N\}$ and the pattern $P$ it creates from the set $B = \{b_1, b_2, ..., b_N\}$. The reader may find Figure \ref{fig:SETH_reduction} helpful. The graph will consist of four sections. We name these according to their function in the reduction: the Selection fan-in, the Selection section, the Post-Selection merge section, and the Synchronization loop. 

We start with the Selection fan-in. Let $2^c$ be the smallest power of $2$ such that $2^c \geq N+1$. The Selection fan-in consists of a complete binary tree with $2^c$ leaves where all paths are directed away from the root. The root is labeled $0$ and the children of every node are labeled $0$ and $1$, respectively.

The Selection section consists of $N+1$ paths. We first define the mappings $f_A$ and $f_B$ from $\{0,1\}$ to sequences of length four as $f_A(0) = 1100$, $f_A(1) = 1111$, $f_B(0) = 0110$, $f_B(1) = 0000$. These mappings have the property that
$d_H(f_A(0), f_B(0)) = d_H(f_A(0), f_B(1)) = d_H(f_A(1), f_B(0)) = 2$ and $d_H(f_A(1), f_B(1)) = 4.$
We make the $i^{th}$ path for $1 \leq i \leq N$ a path of $4(d+1)$ vertices with labels matching the string $f_A(a_i[1]) f_A(a_i[2]) ... f_A(a_i[d]) f_A(0)$. We make the $(N+1)^{th}$ path have $4(d+1)$ vertices and match the string $f_A(0)^{d} f_A(1)$. Let $s_i$ denoted the start vertex of path $i$. We arbitrarily choose $N+1$ leaves, $l_1$, $l_2$,..., $l_{N+1}$, from the Selection fan-in and add the edges $(l_i, s_i)$ for $1\leq i \leq N+1$.

We define the implicit label size as $k = \lceil \log(N + 1)\rceil + 4(d+1)$ and $\ell = k-1$. To construct the Post-selection merge section, we start with $N+1$ length $\ell-1$ paths, each matching the string $2^\ell$. For every path in the Selection section, we add an edge from the last vertex in the path to one of the paths matching $2^\ell$. This is done so that every path matching $2^\ell$ in the Post-selection merge section is connected to exactly one path from the Selection section. Next, we merge two vertices if they have the same implicit label. This is repeated until all vertices in the Post-selection merge section have a unique implicit label.

To construct the Synchronization loop we create a directed cycle with $\ell + 1 = k$ vertices. One of these is labeled with the symbol $3$, and the rest with the symbol $2$. Edges from each ending vertex in the Post-selection Merge section to the vertex labeled $3$ are then added. A final edge from the vertex labeled $3$ to the root of the binary tree in the Selection fan-in completes the graph, which we denote as $D$.

Let  $t = 5d + \lceil \log(N+1) \rceil$. To complete the reduction, we make the pattern 
\begin{align*}
    P = &(2^{\ell}3)^t~2^{\lceil \log (N+1) \rceil} f_B(b_1[1]) f_B(b_1[2]) \hdots  f_B(b_1[d]) f_B(1)\\
    &(2^{\ell}3)^t~2^{\lceil \log (N+1) \rceil} ~f_B(b_2[1]) f_B(b_2[2]) \hdots f_B(b_2[d]) f_B(1)\\&\hdots\\
    &(2^{\ell}3)^t~2^{\lceil \log (N+1) \rceil} f_B(b_N[1]) f_B(b_N[2]) \hdots f_B(b_N[d]) f_B(1)
\end{align*}
and the maximum number of allowed substitutions
$\delta = N \lceil \log_2(N + 1)\rceil + 2(d+1) + (2d + 4)(N-1).$

We call substrings in $P$ of the form $f_B(b_i[1]) f_B(b_i[2]) \hdots  f_B(b_i[d]) f_B(1)$ and paths in $D$ matching strings of the form $f_A(a_i[1]) f_A(a_i[2]) ... f_A(a_i[d]) f_A(0)$ vector gadgets.
Note that $|E| = \mathcal{O}(dN)$ and $m = |P| = \mathcal{O}(d^2 N)$. Hence, an algorithm for approximate matching running in time $\mathcal{O}(m|E|^{1-\varepsilon} + m^{1-\varepsilon}|E|)$ for some $\varepsilon > 0$ would imply an algorithm for OV running in time $d^{\Theta(1)}N^{2-\varepsilon}$. This implies that once the correctness of the reduction has been established, Theorem \ref{thm:SETH_hardness} follows.

\subsection{Proof of Correctness}
\label{sec:SETH__proof_of_correctness}

Proofs of Lemma \ref{lem:seth_is_debruijn} and Lemma \ref{lem:3s_to_3s} are given in Appendix \ref{sec:appendix_seth}.

\begin{restatable}{lemma}{sethIsDebruijn}
\label{lem:seth_is_debruijn}
The graph $D$ is a de Bruijn graph.
\end{restatable}
\vspace{-2mm}

\begin{restatable}{lemma}{sethThreesToThrees}
\label{lem:3s_to_3s}
In an optimal solution, $3$'s in P are matched with $3$'s in $D$.
\end{restatable}

\begin{lemma}
\label{lem:vector_gadgets_to_vector_gadgets}
In an optimal solution, vector gadgets in $P$ are matched with vector gadgets in $D$.
\end{lemma}

\begin{proof}
Suppose otherwise. By Lemma \ref{lem:3s_to_3s}, this can only occur if some vector gadget in $P$ is matched against the Synchronization loop.
This requires at least $4(d+1)$ substitutions. We can instead match the $\lceil \log(N+1) \rceil$ $2$'s preceding the vector gadget in $P$ with the Selection fan-in and the vector gadget in $P$ with the $(N+1)^{th}$ path in the Selection section. Due to $d_H(f_A(0), f_B(0)) = d_H(f_A(0), f_B(1)) = 2$ and $d(f_A(1), f_B(1)) = 4$, this requires $\lceil \log(N+1) \rceil + 2d + 4$ substitutions in $P$. Since, $\log N < d < 2d$ we have $\log N < 2d - 1$. Using that $N$ is some power of $2$, $\lceil \log(N+1) \rceil + 2d + 4 = \log N  + 1 + 2d + 4 < 4d + 4.$ Hence, the cost decreases by matching the vector gadget in $P$ to a vector gadget in $D$ instead.
\end{proof}

\begin{lemma}
If there exists a vector $a \in A$ and $b \in B$ such that $a \cdot b = 0$, then $P$ can be matched to $D$ with at most $\delta$ substitutions.
\end{lemma}

\begin{proof}
Match the vector gadget for $b$ in $P$ with the vector gadget for $a$ in the Selection section of $D$. This costs $2(d+1)$ substitutions. Match the remaining $N-1$ vector gadgets in $P$ with the $(N+1)^{th}$ path in the Selection section, requiring $(2d + 4)(N-1)$ substitutions in total. The total number of substitutions of $2$'s in $P$ to match the Selection fan-in is $N\lceil \log (N+1)\rceil$. Adding these, the total number of substitutions is exactly $\delta$. The synchronization loop can be used for matching all additional symbols in $P$ without any further substitutions.
\end{proof}

\begin{lemma}
\label{lem:delta_implies_orthogonal}
If $P$ can be matched in $D$ with at most $\delta$ substitutions, then there exists vectors $a \in A$ and $b \in B$ such $a \cdot b = 0$.
\end{lemma}

\begin{proof}
By Lemma \ref{lem:vector_gadgets_to_vector_gadgets}, we can assume vector gadgets in $P$ are only matched against vector gadgets in $D$. Suppose that there does not exist a pair of orthogonal vectors $a \in A$ and $b \in B$. Then, which ever vector gadget in $D$ we choose to match a vector gadget in $P$ to, matching the vector gadget requires at least $2d + 4$ substitutions. Hence, the total cost is at least $(2d + 4)N + N\lceil \log(N+1) \rceil  > \delta$, proving the contrapositive of Lemma \ref{lem:delta_implies_orthogonal}.
\end{proof}

\section{Discussion}
\label{sec:discussion}

We leave open several interesting problems. An NP-completeness proof for Problem \ref{prob:sub} on de Bruijn graphs when $k = \mathcal{O}(\log n)$ and the alphabet size is constant is still needed. Additionally, we need to extend these hardness results to when substitutions are allowed in both the graph and the pattern, and when insertions and deletions in some form are allowed in the graph and (or) the pattern. It seems unlikely that adding more types of edit operations would make the problems computationally easier, and we conjecture these variants are NP-complete on de Bruijn graphs as well. It also needs to be determined whether Problem \ref{prob:sub} is NP-complete on de Bruijn graphs with binary alphabets, or whether the SETH-based hardness results hold for Problem \ref{prob:sub_pattern} on binary alphabets.
A practical question is whether these problems are hard for small $\delta$ values on de Bruijn graphs (the problem for general graphs was proven to $W[2]$ hard in terms of $\delta$ in \cite{DBLP:conf/lata/DondiMZ20}). In applications, the allowed error thresholds are quite small. Clearly, the problems are slice-wise-polynomial with respect to $\delta$, i.e., for a constant $\delta$ it is solvable in polynomial time via brute force, but are they fixed-parameter-tractable in $\delta$? The reduction presented here (as well as the reductions in \cite{DBLP:journals/jal/AmirLL00,DBLP:conf/recomb/JainZGA19}) is based on the Hamiltonian cycle problem, where a large $\delta$ value is used. This makes the existence of such a fixed-parameter-tractable algorithm a distinct possibility.

\bibliographystyle{splncs03}
\bibliography{ref}

\section*{\Huge \center{Appendix}}
\setcounter{section}{0}
\renewcommand*{\theHsection}{chX.\the\value{section}}

\section{Missing Proofs in Section \ref{sec:NPC_proof_of_correctness}}
\label{sec:appendix_npc}

\npcIsDebruijn*

\begin{proof}
There are three properties that must be proven:  
(i) Implicit labels are unique, meaning for every implicit label at most one vertex is assigned that label; 
(ii) There are no edges missing, i.e., if the implicit label of $y \in V'$ is $S \alpha$ for some string $S[1,k-1]$ and symbol $\alpha \in \Sigma$, and there exists a vertex $x \in V'$ with implicit label $\beta S[1,k-1]$ for some symbol $\beta \in \Sigma$, then $(x,y) \in E'$; 
(iii) Implicit labels are well-defined, in that every walk of length $k-1$ ending at a vertex $x \in V'$ matches the same string (the implicit label of $x$).
Property (i) holds since after every edge transformation, vertices with the same implicit label are merged, making every implicit label occur at most once. 
For property (ii), consider the completed $D'$ and an arbitrary vertex $y$ on an arbitrary path $\phi((u,v))$. Regarding a possible edge $(x,y) \in E'$, we have the following cases:
\begin{itemize}
    \item Case: the implicit label of $y$ is 
    $
    S\alpha = \enc(L(u)) \$^W \enc(L(u)).
    $
    Then, any potential $x \in V'$ must have an implicit label 
    $
    \beta S = \beta\enc(L(u)) \$^W \enc(L(u))[1,W-1].
    $ 
    However, the only implicit labels created that have a suffix of the form $\enc(L(u)) \$^W \enc(L(u))[1,W-i]$ have a prefix $\#^{W-i}$. This implies that $\beta = \#$, and the edge $(x,y)$ already exists in $E'$ (under the assumption that there are no vertices with in-degree zero in $V$).

    \item Case: the implicit label of $y$ is 
    $
    S\alpha = \enc(L(u))[i,W]\$^W \enc(L(u)) \#^{i-1}$, $1 < i \leq W + 1.
    $ 
    Then, any potential $x$ must have an implicit label
    $
    \beta S = \beta \enc(L(u))[i,W]\$^W \enc(L(u))\#^{i-2}.
    $ 
    Because the only implicit labels with the substring $\$^W \enc(L(u))$ have a prefix consisting of some suffix of $\enc(L(u))$, this implies $\beta = \enc(L(u))[i-1]$, and $(x,y)$ already exists in $E'$.

    \item Case: the implicit label of $y$ is 
    $
    S\alpha = \$^{W-i} \enc(L(u)) \#^{W} \enc(L(v))[1,i]$, $1 \leq i \leq W.
    $ 
    Then, any potential $x$ must have an implicit label 
    $
    \beta S = \beta \$^{W-i}\enc(L(u)) \#^W \enc(L(v))[1,i-1]. 
    $
    In the case $i < W$, $\beta = \$$ and the edge $(x,y)$ already exists in $E'$. In the case where $i = W$, the only implicit label with a suffix of the form $\enc(L(u))\#^W\enc(L(v))[1,W-1]$, has a prefix $\$$, and the edge $(x,y)$ already exists in $E'$.

    \item Case: the implicit label of $y$ is
    $
    S\alpha = \enc(L(u))[i,W] \#^W \enc(L(v)) \$^{i-1}$, $1 < i \leq W + 1.
    $
    Then, any potential $x$ must have an implicit label 
    $
    \beta S = \beta \enc(L(u))[i,W] \#^W \enc(L(v)) \$^{i-2}.
    $
    Because the only implicit labels with the substring $\#^W \enc(L(v))$ have a prefix consisting of some suffix of $\enc(L(u'))$ where the edge $(u',v)$ is in $D$, the edge $(x,y)$ already exists in $E'$. This is an interesting case, as merges can happen, i.e., $\beta \enc(L(u))[i,W] = \enc(L(u'))[i-1,W]$, $u' \neq u$.

    \item Case: the implicit label of $y$ is
    $
    S\alpha = \#^{W-i} \enc(L(v)) \$^W \enc(L(v))[1,i]$, $1\leq i \leq W.
    $
    Then, any potential $x$ must have an implicit label
    $
    \beta S = \beta \#^{W-i} \enc(L(v)) \$^W \enc(L(v))[1,i-1]. 
    $
    For $i < W$, $\beta = \#$ and the edge $(x,y)$ already exists in $E'$. For $i = W$, this is equivalent to the first case.
    
\end{itemize}

We prove (iii) using induction on the number of edges transformed into paths. Our inductive hypothesis (IH) is that prior to an edge being replaced by a path, property (iii) holds for every vertex added to $V'$ thus far.
Let $i$ denote the number of edges transformed. For $i = 1$, all vertices where there exists such a walk ending at them are on the newly created path, and implicit labels are well-defined.

For $i > 1$, we assume the IH holds for all vertices created in the previous $i-1$ steps of transforming edges and merging. First consider a new vertex $x$ that is created by transforming the $i^{th}$ edge $(u_i,v_i)$. Starting with $x = \phi(u_i)$, if $x$ is merged with another transformed vertex $x'$ having the same implicit label, then all length $k-1$ walks ending at $x'$ match this implicit label, and thus the IH holds for $x$ after merging. Using a secondary induction step, we assume the IH holds post-merging for all vertices between $\phi(u_i)$ and $x$ (not including $x$) on $\phi((u_i, v_i))$. Let $x_{prev}$ be the vertex on $\phi((u_i, v_i))$ before $x$. Since all length $k-1$ walks ending at $x_{prev}$ match $x_{prev}$'s implicit label, the length $k-1$ walks obtained by disregarding the vertex at the start of these walks, and adding the vertex $x$ at the end, all match the implicit label of $x$. At the same time, any vertices merged with $x$, by the IH also have the same implicit label and hence the walks ending at them match the implicit label of $x$. Hence, the IH holds for $x$ after merging it with all vertices having the same implicit label. After processing all vertices on $\phi((u_i, v_i))$, we next consider a previously created vertex $x'' \in V'$ not in $\phi((u_i, v_i))$. Consider a newly created walk $W$ of length $k-1$ ending at $x''$ that is due to a vertex merging with vertices in $\phi((u_i, v_i))$. Since all length $k-1$ walks ending at a vertex $z$ in $\phi((u_i, v_i))$ match the same implicit label, when disregarding some number of vertices at the start of a walk that ends at $z$
and appending new vertices, 
the resulting walk $W$ matches the implicit label for $x''$, and the IH continues to hold for $x''$ as well.
\end{proof}

\npcWalkLength*

\begin{proof}

We first define forward distance and backward distance. Let $x, y \in V'$. The forward distance from $x$ to $y$ is defined as the minimum number of edges on any path from $x$ to $y$ (the usual distance in a directed graph). The backward distance from $x$ to $y$ is defined as the minimum number of edges on any path from $y$ to $x$. We say a marked vertex $\phi(u)$ is \emph{backward adjacent} to $x$ if there exists a walk from $\phi(u)$ to $x$ not containing any other marked vertices, and $\phi(v)$ is \emph{forward adjacent} to $x$ if there exists a walk from $x$ to $\phi(v)$ not containing any other marked vertices.

We use induction on the number of edges transformed.  Our inductive hypothesis (IH) will be that the length of all walks that end at and contain only two marked vertices is $4W$. We add to our IH that a vertex $x$ created from an edge transformation having an implicit label of the form:
\begin{enumerate}
    \item $\enc(L(u))[j,W]\$^W \enc(L(u))\#^{j-1}$, $1\leq j \leq W$, has backward distance $j-1$ from $\phi(u)$, which is its only backward adjacent marked vertex, and forward distance $4W-j+1$ from all of its forward adjacent marked vertices; 
    
    \item $\$^{W-j}\enc(L(u))\#^W\enc(L(v)))[1,j]$, $0 \leq j \leq W$, 
    has backward distance $W + j$ from $\phi(u)$, which is its only backward adjacent marked vertex, and forward distance $3W - j$ from all of its forward adjacent marked vertices; 
    
    \item  $\enc(L(u))[j,W]\#^W\enc(L(v)))\$^{j-1}$, $1 \leq j \leq W$,
    has backward distance $2W + j - 1$ from all of its backward adjacent marked vertices, and forward distance $2W - j + 1$ from $\phi(v)$, which is its only forward adjacent marked vertex; 
    
    \item  $\#^{W-j}\enc(L(v)))\$^W\enc(L(v))[1,j]$, $0 \leq j \leq W$,
   has backward distance $3W + j$ from all of its backward adjacent marked vertices, and forward distance $W - j$ from $\phi(v)$, which is its only forward adjacent marked vertex.
    
\end{enumerate}
The base case, $i = 1$, is satisfied since there exists only one such path and all stated properties hold.
Now, for $i > 1$, let $(u_i, v_i)$ be the $i^{th}$ edge transformed. We assume the IH holds for all vertices and walks created in the first $i-1$ edge transformations. First, observe that for any walk ending at, and containing only two previously created marked vertices, for all vertices on this walk the distances from their forward adjacent marked vertices and backward adjacent marked vertices will not be altered unless one of the vertices on this walk is merged with a vertex on $\phi((u_i, v_i))$. Also, all of the stated properties in the IH also hold for $\phi((u_i, v_i))$ prior to merging any vertices.
Now, let $y$ be a vertex on $\phi((u_i, v_i))$. Starting with $y = \phi(u_i)$, and continuing from $\phi(u_i)$ to $\phi(v_i)$, we merge $y$ with existing vertices when their implicit labels match. 
Because the stated distance properties hold for $x$ and $y$ prior to merging, they continue to hold for the vertex created from merging $x$ and $y$ as well. Moreover, for all of the vertices on any walk containing this now merged vertex the distances from its forward adjacent and backward adjacent marked vertices are unaltered. 
Because for every vertex in the new graph, these distances are unaltered, the IH regarding the length of $4W$ for walks containing only two marked vertices continues to hold as well.
\end{proof}

\begin{figure}
\centering
\begin{minipage}{.42\textwidth}
    \centering
    \includegraphics[width=1\textwidth]{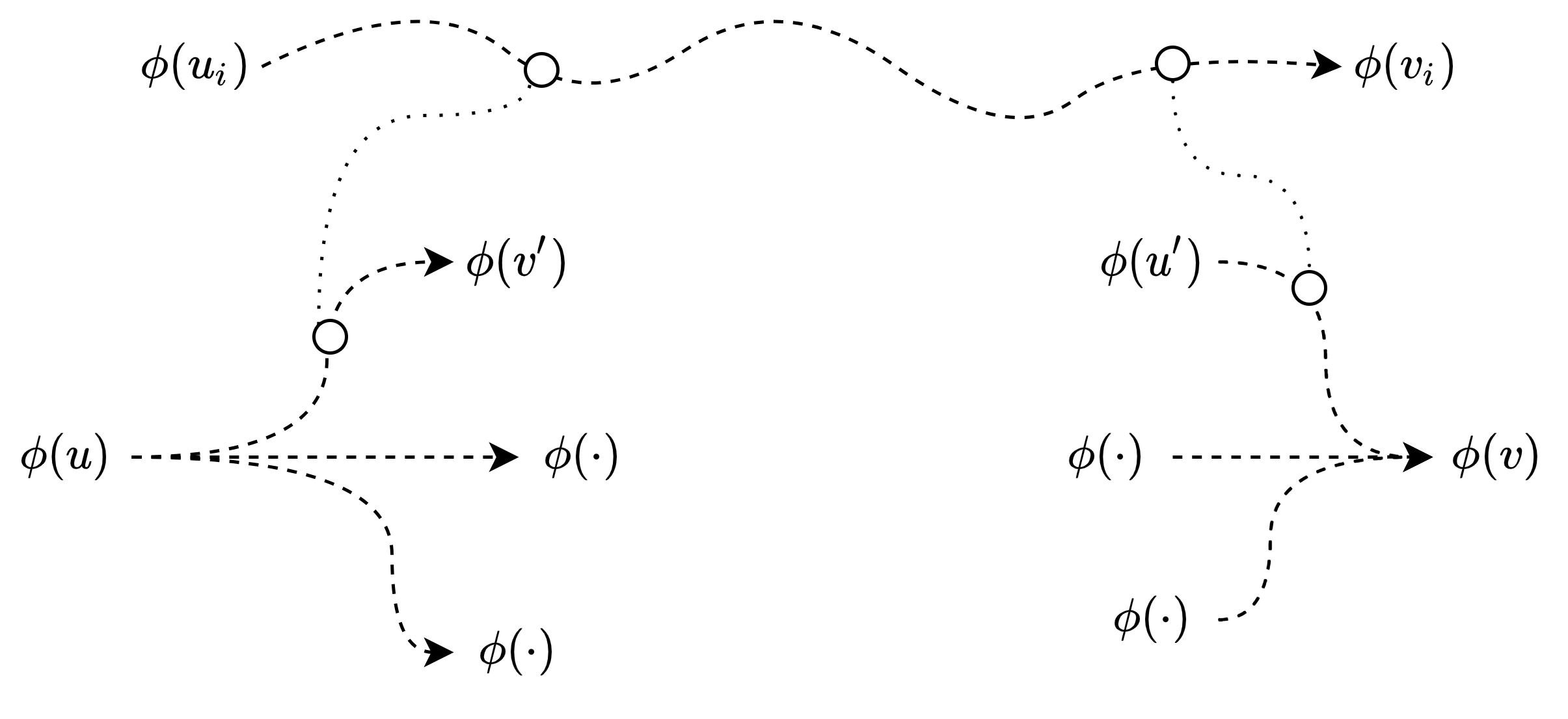}
    \caption{In the proof of Lemma \ref{lem:path_iff_in_D}, we consider whether the path $\phi((u_i, v_i))$ being added could potentially cause a path between $\phi(u)$ and $\phi(v)$. The white circles connected by the thin dashed curve represent merged vertices.}
    \label{fig:path_bridge}
\end{minipage}~~
\begin{minipage}{.58\textwidth}
    \centering
    \includegraphics[width=1\textwidth]{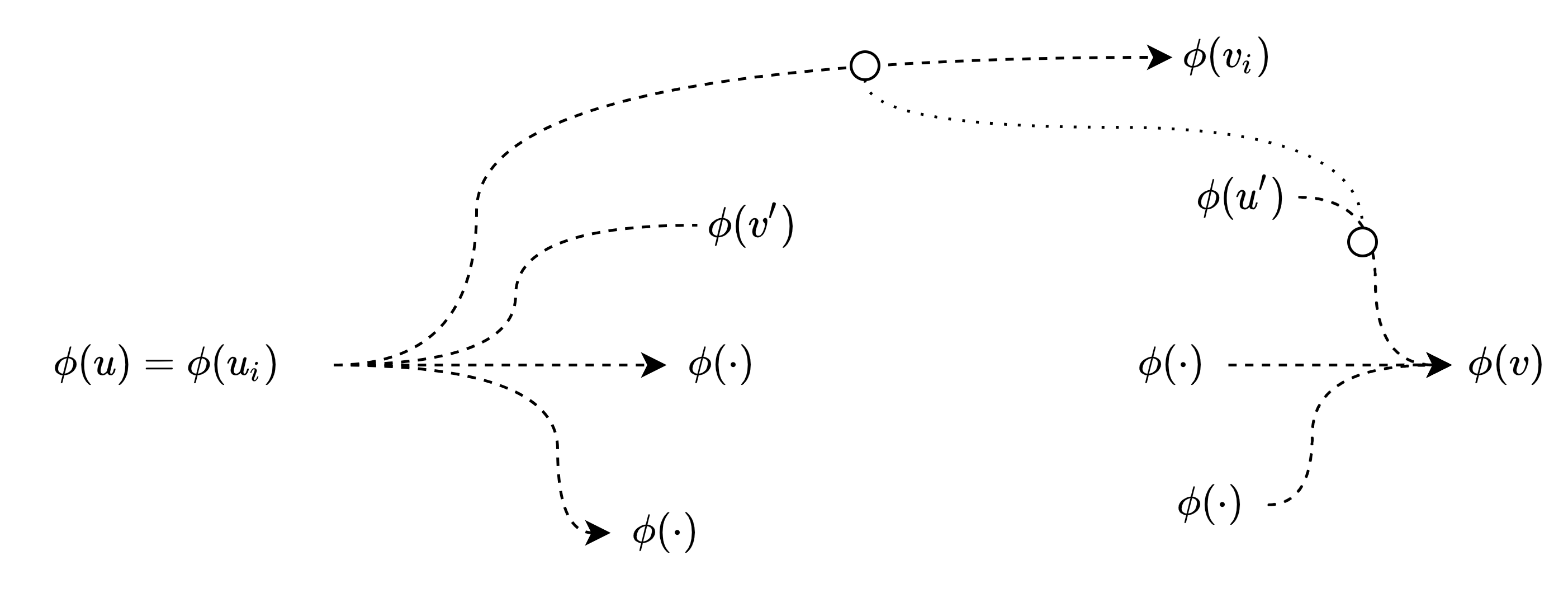}
    \caption{In the proof of Lemma \ref{lem:path_iff_in_D}, the case where $u = u_i$ and $v' \neq v_i$.}
    \label{fig:path_bridge_3}
\end{minipage}
\end{figure}
\vspace{-3em}
\begin{figure}
\centering
\begin{minipage}{.5\textwidth}
    \centering
    \includegraphics[width=\textwidth]{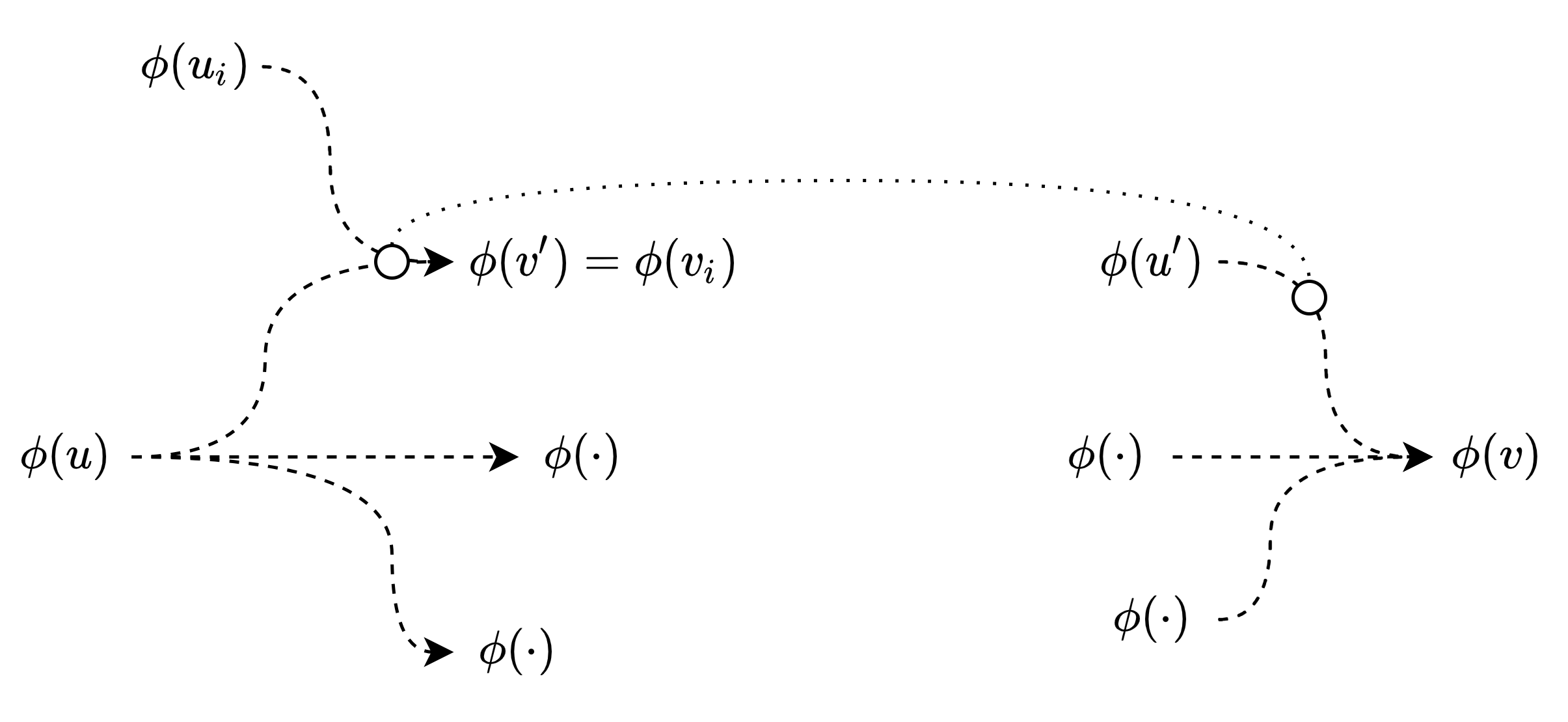}
    \caption{In the proof of Lemma \ref{lem:path_iff_in_D}, the case where $u \neq u_i$, $v' = v_i$, and $u_i \neq u'$.}
    \label{fig:path_bridge_2}
\end{minipage}~~
\begin{minipage}{.5\textwidth}
    \centering
    \includegraphics[width=\textwidth]{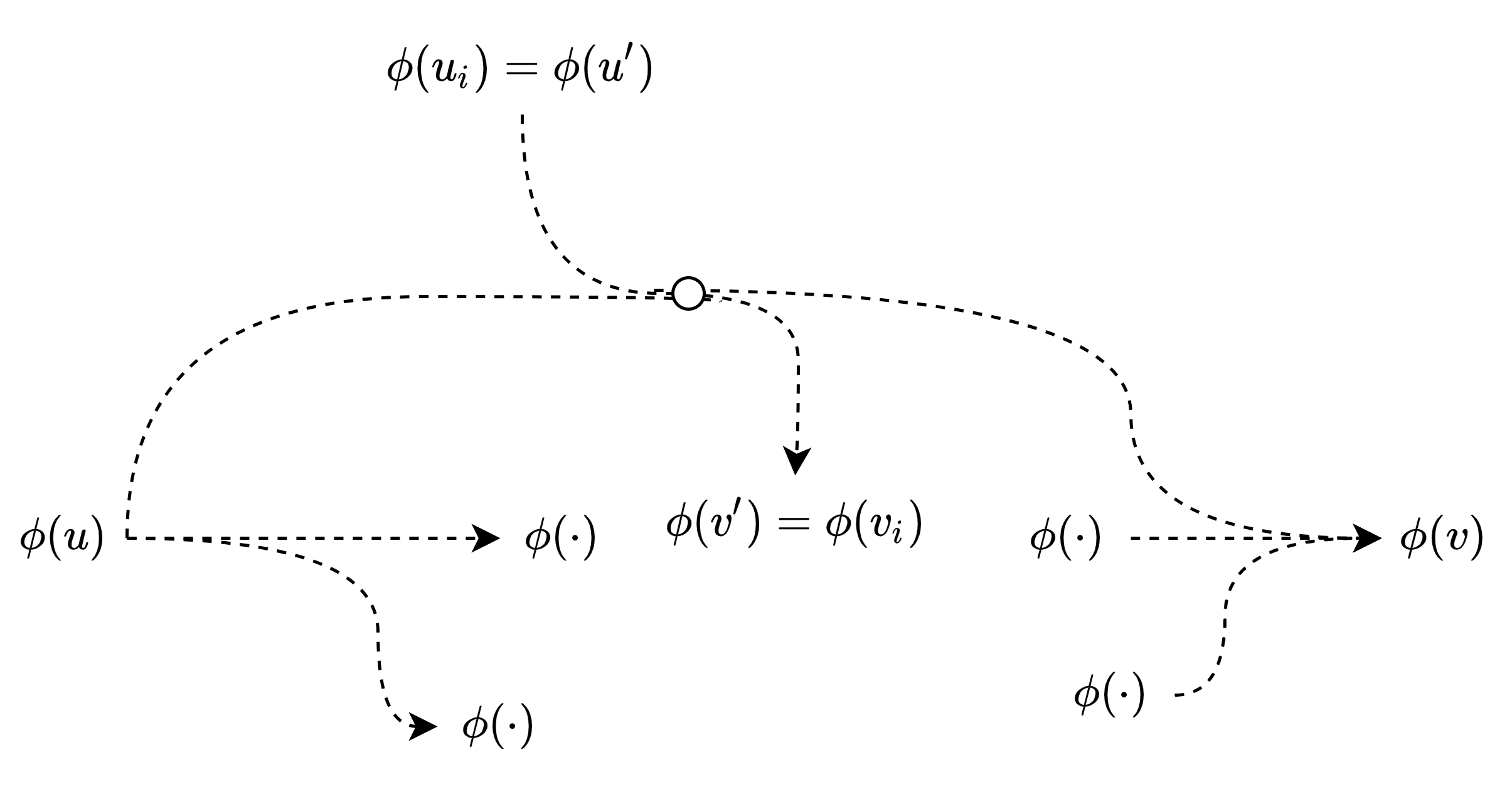}
    \caption{In the proof of Lemma \ref{lem:path_iff_in_D}, the case where $u \neq u_i$, $v' = v_i$, and $u_i = u'$.}
    \label{fig:path_bridge_4}
\end{minipage}   
\end{figure}

\npcPathIffInD*

\begin{proof}
It is clear from construction that if there is an edge $(u,v) \in E$, then such a walk is in $D'$.

In the other direction, suppose for the sake of contradiction that there exists such a walk starting at $\phi(u)$ and ending at $\phi(v)$ with no other marked vertices between $\phi(u)$ and $\phi(v)$ on the walk, and $(u,v) \notin E$. Let the first such walk 
be created when transforming the $i^{th}$ edge $(u_i, v_i)$. The only way such a walk could exists is if some vertex in $\phi((u_i, v_i))$ is merged with a vertex on a walk $\phi((u,v'))$ for some $v' \neq v$, and some vertex in $\phi((u_i, v_i))$ merged with a vertex in a walk $\phi((u',v))$ for some $u' \neq u$. This is since, prior to creating $\phi((u_i, v_i))$ all walks starting at $\phi(u)$ encountered some other marked vertex, say $\phi(v')$, before $\phi(v)$. Similarly, there existed some set of marked vertices not including $\phi(u)$ such that every walk containing a marked vertex and ending at $\phi(v)$ must include at least one vertex in this set, say $\phi(u')$. See Figure \ref{fig:path_bridge}. Consider cases:

\begin{itemize}

\item $u = u_i$ and $v' = v_i$: This contradicts the assumption that $(u_i, v_i)$ is transformed on the $i^{th}$ step.

\item $u = u_i$ and $v' \neq v_i$ (Figure \ref{fig:path_bridge_3}): By Lemma \ref{lem:disjoint_paths}, since $u_i = u \neq u'$, $\phi((u_i, v_i))$ and $\phi((u', v))$ can only share a vertex if $v_i = v$. However, this implies the edge $(u_i, v_i) = (u, v) \in E$, a contradiction.

\item  $u \neq u_i$ and $v' \neq v_i$:  We can directly use Lemma \ref{lem:disjoint_paths} to say no such merged vertices exists between $\phi((u,v'))$ and $\phi((u_i, v_i))$.

\item $u \neq u_i$ and $v' = v_i$ (Figure \ref{fig:path_bridge_2}): By Lemma \ref{lem:disjoint_paths}, if $u_i \neq u'$, then $\phi((u_i, v_i))$ and $\phi((u', v))$ can only share a vertex $v = v_i$. However, this would imply $v = v'$, a contradiction.

The more interesting case is if $u_i = u'$ (Figure \ref{fig:path_bridge_4}). Any vertex $y$ having an implicit label containing $\enc(L(u'))$ and occuring in $\phi((u_i, v_i))$ and $\phi((u',v))$  occurs before (has smaller backward distance to $\phi(u')$) any vertex with implicit label containing $\enc(L(v'))$. At the same time, any vertex $x$ occuring in $\phi((u,v'))$ and $\phi((u_i, v_i))$ has an implicit label containing $\enc(L(v'))$.  Since the vertex $x$ occurs later in $\phi((u_i, v_i))$ than any shared vertex $y$ in $\phi((u_i, v_i))$ and $\phi((u', v))$, the only way any vertices in $\phi((u_i, v_i))$ are in a walk from $\phi(u)$ to $\phi(v)$ not containing any other marked vertices is if there is walk from $x$ to $y$ not containing marked vertices, however, the cycle this creates contradicts Lemma \ref{lem:walk_length}.
\qedhere
\end{itemize}
\end{proof}


\npcDollarSigns*

\begin{proof}
We first make the following observations: pre-substitution of any of the vertex labels in $D'$,

\begin{itemize}

\item (1) For all $u \in V$, there is exactly one path in $D'$ that matches 
$
\enc(L(u)) \#^W \enc(L(u))[1,W-\ell],
$ 
and all vertices on this path have in-degree and out-degree one. This follows from the only vertices with in-degree greater than one having implicit labels 
$
\enc(L(u))[i,W]\#^W\enc(L(v))\$^{i-1}$ where $W-\ell < i \leq W+1$ 
(these vertices have vertex label $\$$). And the vertices with out-degree greater than one having implicit labels of the form 
$
\$^{W-i}\enc(L(u))\#^W \enc(L(v))[1,i]$ where $W-\ell \leq i \leq W$
(the last $\ell$ symbols in $\#^W\enc(L(v)$).
This path contains the marked vertex $\phi(u)$. Furthermore, all marked vertices are included on exactly one such path.

\item (2) Every maximal walk containing only $\$$ or $\#$ symbols is of length $W$, and the distance from the end of any maximal walk consisting of only $\$$ symbols (or $\#$ symbols) to the start of a maximal walk consisting of only $\#$ (or $\$$ symbols resp.) is $W$. This follows from the construction of $D'$: every vertex added in the construction has an implicit label where all maximal substrings consisting of non-$\$$ or non-$\#$ are of length $W$, and maximal substrings consisting of $\$$ or $\#$ are of length $W$. 
\end{itemize}

To see the `local' number of substitutions caused by matching a $\#/\$$-symbol in $D'$ to a $0/1$ symbol in $P$, suppose the matching of $\enc(L(u))$ in $P$ is `shifted left' by $1 \leq s < W$ so that the first $s$ symbols of some $\enc(L(u))$ in $P$ are matched against the last $s$ symbols in some walk of $\$$/$\#$-symbols in $D'$. These last $s$ symbols require $s$ substitutions. In addition, assuming $s < 2\ell$, due to the prefix $(0^{2\ell} 1)^{2\ell}$, at least $2\ell-1$ substitutions that do not involve a $\#$ or $\$$ symbol are needed as well.

We now look at the number of substitutions needed on a `global' level. Using Lemma \ref{lem:walk_length}, it can be inferred that every walk of length $4W$ contains an originally marked vertex. Hence, while matching $P'$ at least $\lfloor|P'|/4W\rfloor = 4Wn/4W  = n$ times an originally marked vertex is visited. Because every substring of $P' = P[1, |P|-4W]$ of length $3W - \ell$ is distinct, every path described in Observation 1 is traversed at most once while matching $P'$. Since each originally marked vertex is on a unique path that can be traversed at most once, and we traverse at least $n$ such paths, we traverse $n$ distinct paths described in Observation 1. 
We can now use Observation 2 to infer that the substitutions needed to match the shifted patterns in $P'$ must be repeated $n$ times.
Hence, to match $P'$ the total number of substitutions involving $\$$/$\#$ symbols is at least $sn$. When $s < 2\ell$, the total number of substitutions is at least
$
(s + 2\ell - 1)n > 2\ell(n-1) = \delta.
$
When $2\ell \leq s < W$, then $2\ell$ substitutions to match the substring $(0^{2\ell} 1)^{2\ell}$ in $P$ may not be needed, but the total number of substitutions required is still greater than $\delta$ since $sn \geq 2\ell n > \delta$.
A symmetric argument can be used for when the matching of $P$ to $D'$ is `shifted right' by $s$ so that the last $s$ symbols in $\enc(L(u))$ in $P$ are matched against the first $s$ symbols in some walk of $\$$/$\#$-symbols in $D'$. 

For $W < s < 4W$, it still holds that all paths described in Observation 1 are traversed exactly once. Combined with Observation 2, we can infer that the substitution cost incurred when making one path of length $W$ originally matching $\#^W$'s match a substring of $P$ without $\#$'s is incurred at least $n$ times. This results in the total number of needed substitutions being at least $nW > \delta$. 
\end{proof}

\section{Missing Proofs in Section \ref{sec:SETH__proof_of_correctness}}
\label{sec:appendix_seth}

\sethIsDebruijn*

\begin{proof}
For each of the four graph sections discussed above, we will prove for each vertex in that section that Conditions (i)-(iii) from the proof of Lemma \ref{lem:is_debruijn} hold. That is, every vertex $v$, $v$'s implicit label well-defined, unique, and there are no additional edges that should have $v$ as their head.

\begin{itemize}
\item Selection fan-in:

\begin{itemize}
\item (well-defined) For any vertex $v$ in the Selection fan-in, there are two paths of length $k-1$ leading to $v$ (one containing vertices labeled with $2$'s from the Post-selection merge section and one containing vertices labeled with $2$'s from the Synchronization loop). Both match the same string $2^{\ell'} 3 B$ where $\ell' < \ell$ and $B$ is a binary string of length at most $\lceil \log N + 1 \rceil$.

\item (unique) The binary string $B$ could only possibly occur again as a suffix the Selection section. However, all implicit labels occurring in that section contain longer binary strings. Hence the implicit label occurs only once in $D$.

\item (no missing inbound edges) Let $u$ be any vertex such that $(u,v)$ is in $D$. A vertex $v$ in the Selection fan-in has an implicit label of the form
$
S \alpha = 2^{\ell'} 3 B_{i'}[1,i]$, $\ell' <\ell$, $1 \leq i < \lceil \log N \rceil$, $0 \leq i' \leq N+1$.
This implies that $u$ has the implicit label
$
\beta S = \beta 2^{\ell'} 3 B_{i'}[1,i-1]. 
$
Based on the limited number of implicit labels present in $D$, it must be that $\beta = 2$, and there exists only one such $u$. Hence, the edge $(u,v)$ already exists.

\end{itemize}

\item Selection section:

\begin{itemize}
\item (well-defined) For a vertex $v$ in the Selection section, there are two length $k-1$ paths leading to $v$ (one with 2's from the Post-selection merge section and one with 2's from the Synchronization loop). Both match a string of the form 
$
2^{\ell'} 3 B_{i'} f_A(a_i[1]) f_A(a_i[2]) ... f_A(a_i[j])[1,h]$ 
where $0 \leq \ell' < \ell$ and $1 \leq h \leq 4$. 

\item (unique) If $v$ has a path of length $k-1$ matching $2^{\ell'} 3 B_{i'} f_A(a_i[1]) f_A(a_i[2])$ $... (f_A(a_i[j])[1,i]$, then it must be in the Selection section. The substring $B_{i'}$ following the prefix $2^{\ell'}3$ is distinct, hence this implicit label only occurs once in the Selection section.

\item (no missing inbound edges) Taking $u$ and $v$ as above, if the vertex $v$ has an implicit label of the form
$
S \alpha = 2^{\ell'} 3 B_{i'} f_A(a_i[1]) f_A(a_i[2]) ... f_A(a_i[j])[1,h]$, $1\leq h \leq 4$,
this implies that the any potential $u$ has an implicit label
$
\beta S = \beta 2^{\ell'} 3 B_{i'}[1,h-1]f_A(a_i[1]) f_A(a_i[2]) ... f_A(a_i[j])[1,h-1]
$
or
$
\beta S = \beta 2^{\ell'} 3 B_{i'}[1,h-1]f_A(a_i[1]) f_A(a_i[2]) ... f_A(a_i[j-1]). 
$
In either case, $\beta = 2$, and the edge $(u,v)$ already exists.
If the vertex $v$ has an implicit label of the form
$
S \alpha = B_{i'} f_A(a_i[1]) f_A(a_i[2]) ... f_A(a_i[d]),
$
then any potential vertex $u$ has an implicit label
$
\beta S = \beta B_{i'} f_A(a_i[1]) f_A(a_i[2]) ... f_A(a_i[d])[1,3]
$
where $\beta$ must be $3$, and the edge $(u,v)$ already exists.
\end{itemize}

\item Post-selection merge section:

\begin{itemize}
\item (well-defined) For a vertex $v$ in this section, all length $k-1$ paths ending at $v$ match a string of the form $B 2^{\ell'}$ where $B$ is a binary string. By construction, the paths ending at $v$ match the same string (they were merged based on this condition).

\item (unique) Again by construction, if another vertex $v'$ in the Post-selection merging section has a length $k-1$ path ending at it that matches $v$'s implicit label $v'$ will be merged with $v$. At the same time, vertices in the other sections of $D$ will not have an implicit label of the form $B2^\ell$.

\item (no missing inbound edges) Taking $u$ and $v$ as above, vertex $v$ has an implicit label of the form
$
S \alpha = B 2^{\ell'}$, $\ell' \geq 1,
$
this implies that any potential vertex $u$ has an implicit label
$
\beta S = \beta B 2^{\ell'-1}.
$
Such a vertex $u$ is already in the Post-selection merge section or is a vertex at the end of a path in the Selection section (if $\ell' = 1$). Since appending a $2$ and removing $\beta$ will make the implicit label of $u$ equal to the implicit label of $v$, the vertex at the head of the edge with tail $u$ must have been merged with $v$. Hence, the edge $(u,v)$ already exists.
\end{itemize}

\item Synchronization loop:

\begin{itemize}
\item (well-defined) There are two length $k-1$ paths to a vertex v in the synchronization loop. Both match the string $2^{\ell'} 3 2^{\ell''}$ where $\ell' + \ell'' = k-1 = \ell$, and $\ell'$ depends on $v$'s position within the Synchronization loop.

\item (unique) An implicit label for a vertex in any other section contains a symbol that is not a $2$ or a $3$. Within the synchronization loop, each implicit label clearly occurs exactly once.

\item (no missing inbound edges) Taking $u$ and $v$ as above, vertex $v$ has an implicit label of the form
$
S \alpha = 2^{\ell'} 3 2^{\ell''}.
$
This implies that any potential vertex $u$ has an implicit label
$
\beta S = \beta 2^{\ell'} 3 2^{\ell''-1}.
$
If $\ell' < \ell$ it must be that $\beta = 2$ and the edge $(u,v)$ already exists. If instead $\ell' = \ell$, then for both
$
\beta S = 0 2^\ell 3
$
and
$
\beta S = 1 2^\ell 3
$
there already exists an edge $(u,v)$ as well. \qedhere
\end{itemize}
\end{itemize}
\end{proof}

\sethThreesToThrees*

\begin{proof}
Suppose that some $3$ in $P$ is not matched with $3$ in $D$ or with the final vertex in a path in the Selection section. Since any walk between $3$'s in $D$ has a length that is a multiple of $k$ and $3$ in $P$ are $k-1$ symbols apart, all $3$'s must then not be matched with $3$ in $D$. This requires at least $tN$ substitutions within $P$. On the other hand, when $3$'s in $P$ are matched with $3$'s in $D$, there exists a solution requiring at most $4d(N+1) +N\lceil \log(N + 1) \rceil$. Specifically, this is obtained by matching each vector gadget in $P$, $f_B(b_i[1]) ... f_B(b_i[d])$ to the $N+1^{th}$ path in the Selection section. Since $t = 5d + \lceil \log(N+1) \rceil > 4d + \frac{4d}{N} + \lceil \log(N+1) \rceil$ for $d = o(N)$ and $N$ large enough, we can assume that $tN > 4d(N+1) +N\lceil \log(N + 1) \rceil$. Hence, all $3$'s in $P$ are matched with the $3$ in $D$ or with some final vertex in a path in the Selection section

Next, suppose some $3$ in $P$ is matched with the last vertex in a path in the Selection section. We consider the first such occurrence. In the case where this occurrence of $3$ in $P$ is followed in $P$ by a substring $2^{\lceil \log(N+1)\rceil}f_B(a_i[1])...f_B(a_i[d])f_B(1)$, a cost of at least $8(d+1)$ is incurred, first at least $4(d+1)$ from matching the substring $2^{\ell} 3$ in $P$ to a path through Selection fan-in and the Selection section, then an additional $4(d+1)$ from matching a vector gadget in $P$ to a path of $2$'s in the Post-selection merge section. We could have instead matched the Synchronization loop twice with a cost of only $4(d+1)$ substitutions, and started and ended at the same vertex while still matching $2^\ell 3 2^{\lceil \log(N+1)\rceil}f_B(a_i[1])...f_B(a_i[d])f_B(1)$. Hence, in this case, matching $3$ in $P$ with the last vertex in a path in the Selection section is suboptimal. In the case where the occurrence of $3$ in $P$ is followed in $P$ by $2^\ell3$, then the cost incurred is only $4(d+1)$. However, we could have instead matched $2^\ell 3 2^\ell 3$ with the Synchronization loop twice with a substitution cost of $0$, and again started and ended at the same vertex. Hence, matching $3$ in $P$ with the last vertex in a path in the Selection section is again suboptimal.
\end{proof}

\end{document}